\documentclass{article}
\usepackage{ijcai24}
\usepackage{times}
\usepackage{soul}
\usepackage{url}
\usepackage[hidelinks]{hyperref}
\usepackage[utf8]{inputenc}
\usepackage[small]{caption}
\usepackage{graphicx}
\usepackage{amsmath}
\usepackage{amsthm}
\usepackage{booktabs}
\usepackage[switch]{lineno}

\newtheorem{theorem}{Theorem}
\newtheorem{lemma}{Lemma}

\author{
Ben Bals$^{1,2}$
\and
Michelle Döring$^{3}$\and
Nicolas Klodt$^{3}$\And
George Skretas$^3$\\
\affiliations
$^1$Centrum Wiskunde \& Informatica, Amsterdam\\
$^2$Vrije Universiteit Amsterdam\\
$^3$Hasso Plattner Institute, Potsdam\\
\emails
bjjb@cwi.nl,
\{michelle.doering, nicolas.klodt, georgios.skretas\}@hpi.de,
}

\title{Catch Me If You Can: Finding the Source of Infections in Temporal Networks} %TODO Please add

\usepackage[disable]{todonotes}
\usepackage{placeins}
\usepackage[normalem]{ulem}
\usepackage[adversary]{cryptocode}
\usepackage{bm}
\usepackage{tabularx}
\usepackage{makecell}
\usepackage{subcaption}
\usepackage{amsmath}
\usepackage{amssymb}
\usepackage{enumitem}
\usepackage{booktabs}
\usepackage{appendix}
\usepackage[appendix=append]{apxproof}
\usepackage{layouts}

\usepackage[ruled,noend]{algorithm2e}

\ifdefined\N\else\DeclareMathOperator{\N}{\mathbb{N}}\fi
\ifdefined\R\else\DeclareMathOperator{\R}{\mathbb{R}}\fi

\newcommand{\BigO}{\mathcal{O}}
\newcommand{\iphase}{\delta}
\newcommand{\nodeset}[1]{V(#1)}
\newcommand{\edgeset}[1]{E(#1)}

\newcommand{\setsize}[1]{\left|#1\right|}
\newcommand{\nodesetsize}[1]{\setsize{V(#1)}}

\newcommand{\Tmax}{T_\mathrm{max}}
\DeclareMathOperator{\edgeLabelOp}{\lambda}
\newcommand{\edgeLabel}[1]{\edgeLabelOp(#1)}
\newcommand{\G}{\mathcal{G}}

\newcommand{\tw}{\mathrm{tw}}
\newcommand{\funnyconstant}{2+\varepsilon}
\newcommand{\EOp}{\mathbb{E}}
\newcommand{\E}[1]{\EOp\left[#1\right]}
\newcommand{\ProbOp}{\mathbb{P}}
\newcommand{\Prob}[1]{\ProbOp\left[#1\right]}

\DeclareMathOperator{\modOp}{mod}

\newcommand{\recheckNico}[1]{\todo[color=green]{@Niko, recheck: #1}}
\newif\ifpaper
\definecolor[named]{benGray}{rgb}{0.61,0.61,0.63}

\setlist[itemize]{noitemsep, nolistsep}

\papertrue

\newtheoremrep{appendixtheorem}[theorem]{$\star$ Theorem}
\newtheoremrep{appendixlemma}[lemma]{$\star$ Lemma}
\newtheoremrep{appendixcorollary}[corollary]{$\star$ Corollary}

\usepackage[capitalise]{cleveref}
\crefname{hypothesis}{Hypothesis}{Hypotheses}
\crefname{conjecture}{Conjecture}{Conjectures}

\begin{document}

\maketitle
\begin{abstract}
Source detection (SD) is the task of finding the origin of a spreading process in a network.
Algorithms for SD help us combat diseases, misinformation, pollution, and more, and have been studied by physicians, physicists, sociologists, and computer scientists.
The field has received considerable attention and been analyzed in many settings (e.g.,~under different models of spreading processes), yet all previous work shares the same assumption that the network the spreading process takes place in has the same structure at every point in time.
For example, if we consider how a disease spreads through a population, it is unrealistic to assume that two people can either never or at every time infect each other, rather such an infection is possible precisely when they meet.
Therefore, we propose an extended model of SD based on temporal graphs, where each link between two nodes is only present at some time step.
Temporal graphs have become a standard model of time-varying graphs, and, recently, researchers have begun to study infection problems (such as influence maximization) on temporal graphs \cite{influencers,gayraud-evolving-social-networks}.
We give the first formalization of SD on temporal graphs.
For this, we employ the standard SIR model of spreading processes \cite{Hethcote1989}.
We give both lower bounds and algorithms for the SD problem in a number of different settings, such as with consistent or dynamic source behavior and on general graphs as well as on trees.
% Crucially, we show that in the worst case, the trivial brute force algorithm is always optimal and therefore propose randomized algorithms that succeed with constant probability.
% For these, we prove matching lower bounds to show that they are asymptotically optimal.
\end{abstract}

\section{Introduction}
\label{intro}

Information, diseases, pollutants—all these things spread through the networks that define how everything from people and servers to our streets and sewers are connected. The field of spreading processes has received substantial attention from researchers as diverse as physicians, physicists, sociologists, and computer scientists. 
One of the central questions that naturally arises is to find the source of some spreading process given only limited information (e.g.,~from a small set of disease infection times or from a number of sensors in a sewer network) \cite{sir-source-detection,shah2011rumors,sensor-placement,paluch_fast_2018}.
The current mainstream models in this area, such as the independent cascade model \cite{independent-cascade}, or the susceptible-infected-resistant (SIR) model \cite{Hethcote1989,a-contribution,shah2011rumors} incorporate the inherent temporal behavior of the spreading process, that is, infections themselves take time such that a node first gets infected and only then is able to infect other nodes. However, they fail to capture temporal dependencies in the underlying network itself. In many real-world networks the connections change over time (e.g.,~social networks in which diseases or information spread are highly temporal as the link between two people is only able to transmit information or diseases at specific points in time) and these changes heavily impact the behavior of spreading processes. For this reason, the study of infections in temporal graphs has been initiated. In 2015,
Gayraud et al.~\cite{gayraud-evolving-social-networks}, generalized the independent model from static to temporal graphs. Since then, other papers have studied problems related to spreading processes in temporal networks \cite{influencers}, but, to the best of our knowledge, no one has conducted research on the problem of \emph{source detection} (SD). We provide the first rigorous definition of SD in temporal graphs and give both algorithms and lower bounds for this problem in a variety of settings.\footnotemark%
\footnotetext{This work is based on the first author's master thesis, which is available at \url{https://arxiv.org/abs/2503.13567}.}

\subsection{Our Contribution}
\label{contribution}

\begin{table*}[t]
\caption{\label{tbl:results}An overview of the price of detection in different settings. For lower bounds, a \emph{wcp} annotation signifies that the bound holds for all Discoverer algorithms winning the game with constant probability. For upper bounds, \emph{wcp} signifies that a Discoverer algorithm winning the game with constant probability and the noted price of detection exists. Upper bounds marked \emph{det} are achieved via a deterministic algorithm. Results marked \textsuperscript{a} are transferred between known and unknown settings (i.e.,~lower bounds from known to unknown and upper bounds in the other direction). Similarly, results transferred between the consistent and obliviously dynamic settings are marked \textsuperscript{b} and transfers between trees and general graphs are marked \textsuperscript{c}. The result marked * only holds when the Discoverer is allowed to watch two nodes.}
\centering
\begin{tabular}{lllll}
\toprule
& \multicolumn{2}{c}{\textbf{Trees}}  & \multicolumn{2}{c}{\textbf{General}} \\[0pt]
\cmidrule(r){2-3} \cmidrule(l){4-5}
& Lower Bound & Upper Bound & Lower Bound & Upper Bound\\[0pt]
\midrule
\textbf{Consistent} &  &  &  & \\[0pt]
Known & \(\Omega(n \log n)\) wcp & \(\BigO(n \log n)\) wcp & \(\Omega(n \sqrt {n})\) wcp &  \textcolor{benGray}{\(\BigO(n \sqrt{n})\) wcp\textsuperscript{a}}\\[0pt]
& \Cref{thm:aX1-lb-tree} & \Cref{cor:aX1-tree} &  \textcolor{benGray}{\Cref{thm:consistent-known-to-unknown}} & \\[0pt]
Unknown & \(\Omega(n \sqrt n)\) wcp &  \textcolor{benGray}{\(\BigO(n \sqrt n)\) wcp\textsuperscript{c}} &  \textcolor{benGray}{\(\Omega(n \sqrt n)\) wcp\textsuperscript{c}} & \(\BigO(n \sqrt {n})\) wcp\\[0pt]
& \Cref{thm:aY1-lb} &  &  & \Cref{thm:aY1}\\[0pt]
\midrule
\textbf{Obliviously dynamic} &  &  &  & \\[0pt]
Known &  \textcolor{benGray}{\(\Omega(n \log n)\) wcp\textsuperscript{b}} & \(\BigO(n \log n)\) wcp* & \(\Omega(n^2)\) wcp &  \textcolor{benGray}{\(n^2\) det\textsuperscript{a}}\\[0pt]
&  & \Cref{thm:dX2} & \Cref{thm:dX1-lb} & \\[0pt]
Unknown &  &  &  \textcolor{benGray}{\(\Omega(n^2)\) wcp\textsuperscript{a}} & \(n^2\) det\\[0pt]
&  &  &  & \Cref{thm:dY1}\\[0pt]
\bottomrule
\end{tabular}
\end{table*}

Our contribution is twofold: (i) we define the SD problem as a round-based, interactive two-player game, (ii) we provide algorithms and lower bounds for different parameters of the SD game, and we provide theoretical proofs for all of our claims. See \Cref{tbl:results} for an overview of our results. 
\noindent {\emph{Statements where proofs or details are omitted due to
space constraints are marked with $\star$. An appendix
containing all proofs and details is provided as supplementary material.}}

In \Cref{sec:game-def}, we formalize the SD problem as a round-based game in which two players interact with each other.
The \emph{Adversary} decides the structure of the underlying temporal network, as well as which node is the source and when it starts an infection chain.
The \emph{Discoverer} initially has no information about the location, number, or label of the edges.
In each round, the Discoverer may watch a single node and is informed if, when, and by which neighbor this watched node is infected.
Note that the Adversary does not know which nodes the Discoverer watches.
Counting the rounds of this game until the Discoverer has found the source is an insufficient cost metric, as then the Adversary is incentivized to minimize infections in order to evade detection.
Instead, we study the number of infections until detection, which we call the \emph{price of detection}.
This cost measure combines an incentive to avoid detection (to prolong the period in which the Adversary may perform infections) and an incentive to infect as many nodes as possible in a single round.
Notice, that if \(n\) is the number of nodes in the graph, any Discoverer algorithm must always tolerate \(\Omega(n^2)\) infections in the worst case, and there is an algorithm that always wins the game within \(n^2\) infections (see \Cref{thm:dY1}).

Because in this game, the worst-case performance of any algorithm is trivially bad \todo{Is this true? Half of our results are on trees. Is there a trivial reason why you cannot do it faster on trees?} (i.e.,~its price of detection is in \(\Omega(n^2)\)), we explore the power of randomized techniques to overcome this limitation in \Cref{sec:rand}.
In particular, we give a Discoverer algorithm that wins the SD game with constant probability (i.e.,~a probability that does not decrease for larger networks) while only tolerating \(\BigO(n \sqrt n)\) infections.
We also prove that this price of detection is asymptotically optimal for any algorithm that wins the game with a constant probability.

In \Cref{sec:known}, we explore how the problem changes if the Discoverer has knowledge about the underlying static graph (but not about when an edge exists in the temporal graph). Surprisingly, we prove that this does not lead to a decreased price of detection. In particular, we show that the lower bounds from the unknown setting transfer to the known one. Despite this, we can show positive results on structured instances. We give an algorithm that, for graphs with bounded treewidth \(\tw\), wins the game on known static graphs with constant probability with a price of detection of \(\BigO(\tw \cdot n \log n)\).
This directly translates to a \(\BigO(n \log n)\) algorithm on trees, which we discuss in \Cref{sec:trees}.

%The results up until now assume that the behavior of the source is consistent, that is, the source begins its infection chain at the same point in time in each round.
%While this is not true in all real-world settings, it is equivalent to taking multiple measurements in one iteration of a process.
In \Cref{sec:dyn}, we study what happens when we allow the Adversary to change the time-step at which it infects the source between each round. 
Note that we still assume the Adversary does not know which nodes the Discoverer chooses to watch. That is, while the source behavior may be dynamic, it may not depend on the actions of the Discoverer, thus, we call this model \emph{obliviously dynamic}.
We first show that any Discoverer algorithm that wins the game with constant probability must have a price of detection in \(\Omega(n^2)\).
On a more positive note, we show that if we strengthen the Discoverer slightly by allowing them to watch two nodes each round, we are still able to achieve a price of detection of \(O(n \log n)\) with constant probability on trees (if the static graph is known).
We finish off this line of inquiry by giving a proof that allowing the Discoverer to watch \(k\) nodes may only decrease the price of detection by at most a factor of \(k\) if the source behavior is consistent, thus this generalization may only be asymptotically different in the obliviously dynamic setting.

%Note that most of our results are tight.
%That is, for most settings, we are able to provide an asymptotically optimal Discoverer algorithm that wins the game with constant probability.

\subsection{Related Work}
\label{sec:org42bc633}
The SIR model and related models are standard in mathematical biology \cite{Hethcote1989} and have been extensively studied  from a statistical perspective \cite{BRITTON201024}.
The SIR model most closely models how viral infections spread through populations.
For example, adaptations of the model have been used to study COVID-19 \cite{KUDRYASHOV2021466,COOPER2020110057,covid-sir}.

Recently, researchers have started to ask more algorithmic questions about infection models, aiming to answer questions such as SD in addition to the classical task of predicting spread.
In particular, \cite{independent-cascade}
set off this new wave of algorithmic study of infections.
They introduce and study the \emph{influence maximization problem} (though under the independent cascade model instead of SIR), which already incorporates the inherently temporal nature of spreading infections while the underlying graph is modeled as static.
Their work finds broad resonance with applications as diverse as the study of misinformation \cite{budak-limiting,fake-news-survey}, infectious diseases \cite{covid-sir,computer-immunology}, and viral marketing \cite{marketing-im-billion,10.1145/2896377.2901462}.

The SD problem has a rich literature and has even been studied on the SIR model we employ.
In 2011, \cite{shah2011rumors} study the SD problem under the SIR model.
In 2016, \cite{sir-source-detection} extend their work.
Other notable algorithmic contributions include the PVTA algorithm proposed by \cite{pinto2012locating},
which focuses on estimating the source of an infection from observations on a sparse set of nodes.
This model has also recently been extended to SD \cite{pmlr-v216-berenbrink23a}.
In 2011, \cite{shah2011rumors} study the SD problem under the SI model.
\cite{sir-source-detection} extend this work to the SIR model.\todo{Those exact two papers have been mentioned before with the same text except for that SI was SIR before}
\cite{gayraud-evolving-social-networks}, extend the study of influence maximization on both the independent cascade and linear threshold models to temporal graphs (which were popularized by \cite{kempe2000connectivity}).
\cite{influencers}
study a number of variations of the influence maximization problem on temporal graphs under the SIS model (which is closely related to SIR).
\todo[color=pink]{@Supervisors: Discuss temporal graphs themselves in the related work section?}

\section{Preliminaries}
\label{sec:org6a8d1a2}
We write \(2\N\) for the set of even natural numbers and \(2\N + 1\) for the set of odd natural numbers.
% For \(a \in \N, m \in \N^+\), set \(a \divOp m \coloneqq \lfloor a / m \rfloor\).
% Note that \(a = (a \divOp m)m + a \modOp m\) and thus (for a fixed \(m\)) every integer has precisely one unique representation by \(\divOp\) and \(\modOp\).
Two real-valued functions \(f, g \colon \N \to \R\) are asymptotically equivalent if \((f(n) / g(n)) \xrightarrow{n\to \infty} 1\). We use the standard definition of \emph{treewidth} as presented in \cite{cygan2015parameterized}.

We define a \emph{temporal graph} \(\G = (V, E, \edgeLabelOp)\) with lifetime \(\Tmax\) as an undirected graph \((V, E)\) together with a \emph{labeling function} \(\edgeLabelOp \colon E \to 2^{[\Tmax]}\).
We interpret this as \(e \in E\) being present precisely at time steps \(\edgeLabel{e}\).
We restrict ourselves to \emph{simple} temporal graphs where each edge has exactly one label.
Abusing notation, we therefore also use \(\edgeLabelOp\) as if it were defined as \(\edgeLabelOp\colon E \to [\Tmax]\).
% Until \Cref{sec:multigraphs}, we restrict ourselves to \emph{simple} temporal graphs where each edge has exactly one label.
% Abusing  notation, we therefore also use \(\edgeLabelOp\) as if it were defined as \(\edgeLabelOp\colon E \to [\Tmax]\). In \Cref{sec:multigraphs}, we investigate what happens if we allow an edge to be present at more than one time step and see that many of our results translate.
We also write \(\nodeset{\G}\) for the nodes of \(\G\), and \(\edgeset{\G}\) for its edges.
Every temporal graph has an \emph{(underlying) static graph} \(G=(V, E)\) obtained by simply forgetting the edge labels.
A sequence of nodes \(v_1, \dots, v_\ell \in V\) is a \emph{temporal path} if it forms a path in $G$ and the labels along the edges are strictly increasing (i.\,e.,~for all \(i \in [\ell -2]\), we have \(\edgeLabel{v_i v_{i+1}} < \edgeLabel{v_{i+1} v_{i+2}}\)).
% For the rest of this section, let \(\G=(V, E, \edgeLabelOp{})\) be a temporal graph with lifetime \(\Tmax\).
%In \Cref{sec:multigraphs}, we investigate what happens if we allow an edge to be present at more than one time step and see that many of our results translate.

% Our model of temporal infection behavior is based on \cite{influencers} (and more historically flows from \cite{a-contribution} and \cite{pastor2015epidemic}).
We use the susceptible-infected-resistant (SIR) model, in which a node is either in a \emph{susceptible}, \emph{infected}, or \emph{resistant} state. %, where all nodes start in a \emph{susceptible} state, may get \emph{infected}, remain infectious (i.\,e.,~attempt to infect their neighbors) for a fixed period of time \(\iphase\), and then become \emph{resistant} (i.\,e.,~cannot be infected for the rest of the lifetime).
This model of temporal infection behavior is based on \cite{influencers} (and more historically flows from \cite{a-contribution} and \cite{pastor2015epidemic}).
An \emph{infection chain} in the SIR model unfolds as follows.
At most \(k \in \N\) nodes may be infected by the Discoverer at arbitrary points in time, which we call \emph{seed infections} denoted as \(S \subseteq V \times [0,\Tmax]\). % instead of their neighbors.
% We call these \emph{seed infections} and they may occur at arbitrary points in time.
% Formally, we denote the set of seed infections as \(S \subseteq V \times [0,\Tmax]\).
Otherwise, a node \(u\) becomes \emph{infected at time step \(t\)} if and only if it is susceptible and there is a node \(v\) infectious at time step \(t\) with an edge \(uv\) with label \(t\).
Then \(u\) is infectious from time \(t+1\) until \(t+\iphase{}\), after which \(u\) becomes resistant.
Note that if a susceptible node has two or more infected neighbors at the same time,  it can be infected by any one of them, but only one.
Thus, a given set of seed infections may result in multiple possible infection chains.

\section{The Source Detection Game} \label{sec:game-def}
\label{game-def}
In our game, two players interact with each other. The \emph{Discoverer} attempts to find a fixed source of infections (e.g.,~a bot spreading misinformation) and the \emph{Adversary} controls the environment as well as the position of that source and attempts to conceal the source from the Discoverer.
The goal of the Discoverer is to find the source while minimizing the number of successful infections until they find the source.
We refer to this number of infections as the \emph{price of detection}.
See \Cref{game:source-detection} for a formalization of the game.

% \begin{figure*}[tbph]
% \rule{\textwidth}{0.4pt}\vspace{-0.75\baselineskip} %
% \procedureblock[width=\textwidth]{Source Detection Game}{%
% \textbf{Adversary} \< \< \textbf{Discoverer} \\
% \text{Pick } V, \Tmax, \iphase \< \sendmessageright*{V, \Tmax, \iphase} \< \pclb
% \pcintertext[dotted]{\quad Rounds \(i=0,1,\dots\) \quad}  \\
% \text{Select some infection times} \\ I_i \colon V \to ([0,\Tmax] + \bot)  \< \< \\
% \< \sendmessageleft*{S_i} \< \text{Pick } S_i \subseteq V \text{ with  } |S_i| = k \\
%   \< \sendmessageright*{I_i(S_i)} \< \pclb
% \pcintertext[dotted]{\quad Discoverer decides to end the game \quad}  \\
% \< \sendmessageleft*{s} \< \text{Select suspected source } s \in V \\
% \text{Pick labeling } \edgeLabelOp, \text{ source } s \in V, \\ \text{and time } t_0 \in \Tmax \< \< \\
% \< \sendmessageright*{\edgeLabelOp, s', t_0} \< \pclb
% \pcintertext[dotted]{\quad End of game \quad}
% }
% The Adversary wins if \(\G=(V, E, \edgeLabelOp)\) with parameters \(\iphase\) and \(\Tmax\) as well as a seed infection only at \((s', t_0)\) is consistent with all \(I_i(S_i)\) and \(s \ne s'\). Otherwise, the Discoverer wins.

% \vspace{-0.25\baselineskip} \rule{\textwidth}{0.4pt}
% \caption{The source detection game in the variation with consistent source behavior and known static graph. \label{game:source-detection}}
% \end{figure*}
\todo[color=pink]{Fix that the Discoverer learns the infection source}

\begin{figure}[tbph]
\includegraphics[width=\linewidth]{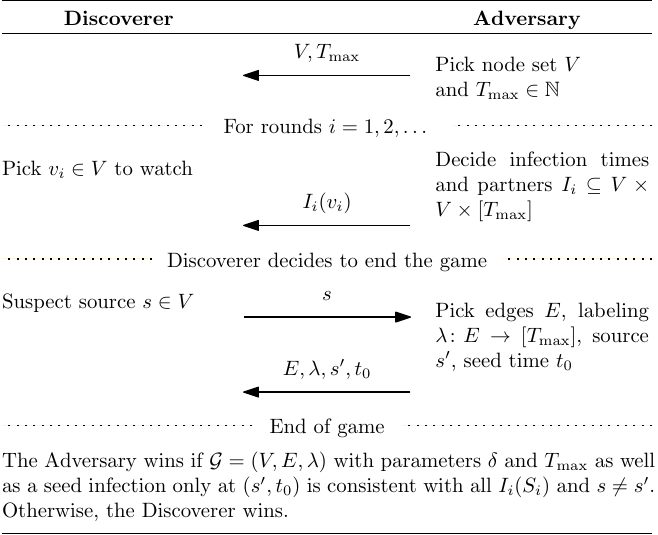}
\caption{The SD game in the variation with consistent source behavior and known static graph. \label{game:source-detection}}
\end{figure}

We investigate different restrictions placed upon the knowledge and power of the Discoverer and Adversary.
The explored dimensions that impact on the price of detection are:
\begin{itemize}
\item \emph{Infection behavior}: Does the source have to infect in the same fashion each round, or can it change its behavior between rounds? In any case, the Adversary is oblivious to which nodes the Discoverer watches.
\item \emph{Watched nodes}: How many nodes may the Discoverer watch in each round? We investigate the cases where the Discoverer watches one or two nodes.
\item \emph{Discoverer knowledge}: Which information about the graph is available to the Discoverer? We investigate the case where the Discoverer knows the underlying static graph and the case where it only knows the nodes of the graph but not the edges.
\item \emph{Graph class}: How does the structure of the underlying static graph affects the price of detection? We look at general graphs, trees, and graphs of bounded treewidth.
\end{itemize}

% \noindent\textbf{Infection behavior}: Does the source have to infect in the same fashion each round, or can it change its behavior between rounds? In any case, the Adversary is oblivious to which nodes the Discoverer watches.\\
% \noindent\textbf{Watched nodes}: How many nodes may the Discoverer watch in each round? We investigate the cases where the Discoverer watches one or two nodes.\\
% \noindent\textbf{Discoverer knowledge}: Which information about the graph is available to the Discoverer? We investigate the case where the Discoverer knows the underlying static graph and the case where it only knows the nodes of the graph but not the edges.\\
% \noindent\textbf{Graph class}: We examine how the structure of the underlying static graph affects the price of detection. Concretely, we look at general graphs, trees, and graphs of bounded treewidth.\\

% Furthermore, for many of these problems, we investigate their price of detection where the Discoverer is a randomized algorithm.
% There, we aim for bounds that hold with constant probability.
% See \Cref{tbl:results} for an overview of our results.

\section{Randomized Source Detection} \label{sec:rand}
\label{rand}
Randomizing the behavior of the Discoverer enables us to evade the worst cases and yield good results with high probability.\todo[color=pink]{@Supervisors: Should I include a very bad lower bound on worst-case performance to justify this? --- Yes I think that would be very helpful.}
In fact, in this section, we even show that there is a randomized Discoverer that wins the SD game with consistent source behavior on unknown static graphs within \(O(n \sqrt n)\) infections and with constant success probability.

First, we give a useful primitive to use in the construction of more complex algorithms.
We essentially prove that the source can, in expectation, not hide too many infections from the Discoverer.
Looking at this fact the other way round, we see that if there is a linear number of infections, a Discoverer employing this strategy likely learns of at least one infection.
Interestingly, this result even holds for all settings.

\begin{lemma}
Consider an instance of the SD game (either with known or unknown graph and either consistent or obliviously dynamic source behavior).
Then there is a strategy for the Discoverer such that, after termination of the strategy, the Discoverer has observed a node in a round in which it gets infected.
This strategy succeeds after tolerating at most \(3n\) infections in expectation.
\label{lem:linear-discovery}
\end{lemma}

\begin{proof}
In each round, the strategy picks a node to watch uniformly at random.
It terminates when it has observed a node which gets infected while being watched.
Let \(T\) be the random variable that takes the number of the round in which this happens.
Let \(a_1, \dots\) be the nodes infected in the respective rounds.
Then the expected number of tolerated infections is
%\begin{align*}
$\E{\sum_{i=1}^T a_i}
= \sum_{i = 1}^\infty a_i \cdot \Prob{i \le T}
= \sum_{i = 1}^\infty a_i \cdot \prod_{j=1}^{i-1} \left(\frac{n-a_j}{n}\right)$,
%\end{align*}
where the equalities follow from the alternative definition of the expectation and since we require at least \(i\) rounds iff the first \(i-1\), rounds are unsuccessful. Now, substitute \(p_i = a_i / n\) for all \(i\). Then \(p_i\) is the probability of finishing in round \(i\) assuming it is reached. We may also pull \(n\) outside the sum to view the number of tolerated infections in each round as a multiple of \(n\). This leaves us with $n \sum_{i = 1}^\infty p_i \cdot \prod_{j=1}^{i-1} (1-p_j)$. Now group the rounds of our strategy and thus the \(p_i\) by taking consecutive elements until their sum is at least \(1/2\) and then start the next group. Formally, let \(\{k_j\}_{j \in \N}\) be the (possibly finite) sequence such that \(k_j\) is the smallest integer such that \(\sum_{i = k_{j}}^{k_{j+1}-1} p_i \ge 1/2\). Then, by union bound, the probability of finishing within a given group, assuming we reach it, is at least \(1/2\). Also, we have that \(\sum_{i = k_j}^{k_{j+1}-1} p_i \le 3/2\) since all but the last elements together are less than \(1/2\) and the last element is at most 1. This is a bound on the tolerated infections (as a multiple of \(n\)). Therefore, we can bound the above expectation from above by the geometric process that tolerates \(3n/2\) infections to perform the rounds in a given group and has probability \(1/2\) of finishing within that block. Thus, we can bound above by
\begin{align*}
\le~ &n \sum_{j=1}^\infty \left( \sum_{i=k_j}^{i=k_{j+1}-1} p_i \right) \cdot \left(\prod_{\ell=1}^{j-1}\left( 1 - \sum_{i=k_\ell}^{i=k_{\ell+1}-1} p_i \right) \right) \\  \le &n \sum_{j=1}^\infty 3/2 \cdot 1/2^{j-1} = 3n.\qedhere
\end{align*}
\end{proof}

By Markov's inequality, the probability that this strategy takes at most \(6n\) infections is at least constant (\(1/2\) to be precise).
Note that in the case of consistent source behavior, this yields a node sampled uniformly from the set of nodes which are infected as we employ uniform rejection sampling.

\begin{algorithm}[tbhp]
\DontPrintSemicolon
\SetKwFunction{FMain}{UnknownSourceDetection}
\SetKwProg{Fn}{fun}{:}{}
           1. Find $\sqrt{n}$ nodes which are infected using repeated application of \Cref{lem:linear-discovery}.

           2. Now let $a$ be the node among those picked that is infected the earliest.

           3. Recursively, pick the next $a$ as the node that infected the previous $a$.

           4. If there is no neighbor left that gets infected earlier, we must have found the root.
    \caption{The randomized SD algorithm for unknown static graphs and consistent source behavior.}
    \label{alg:aY1}
\end{algorithm}

Notice that this algorithm always finds the source and has a constant probability of terminating within \(2 \sqrt n\) rounds. Modifying step 3 such that it aborts after \(\sqrt{n}\) rounds yields an algorithm that always terminates within \(2\sqrt n\) rounds and has a constant probability of finding the source.

With this tool in hand, we can now give \cref{alg:aY1}.

\begin{appendixtheoremrep}
\Cref{alg:aY1} solves the SD problem with consistent source behavior on unknown static graphs within \(O(n \sqrt n)\) infections with constant probability.
\label{thm:aY1}
\end{appendixtheoremrep}

\begin{proofsketch}
We first show that, picking a node close enough to the source in Step 1, results in finding the source in Steps 3 and 4 in \(\BigO(n \sqrt n )\) rounds by examining the tree of the infection behavior.
We then bound the probability of picking a node close enough to the source by showing that not too many nodes can be far away from the source (since being far away requires a number of intermediate nodes).
\end{proofsketch}

\begin{appendixproof}
The proof consists of two parts: first, we show that if we pick a good node as \(a\) in step 2, then steps 3 and 4 take \(O(\sqrt{n})\) rounds. Second, we show that picking a good node is sufficiently likely.

Let \(T\) be the tree of the infection behavior.
Observe that when node-labeled with the infection times, \(T\) is a min-heap.
For a node, call the distance from the root its \emph{infection distance}.
For part one, assume that \(a\) has an infection distance at most \(\sqrt{n}\).
Clearly, then step 3 recurses at most \(\sqrt n\) times.

For part two, set \(t\) to be the smallest time step at which any node \(v\) of infection distance greater than \(\sqrt n\)  gets infected, and let \(T'\) be the induced subgraph on \(T\) of the nodes that are infected before \(t\). Since \(T\) is a min-heap with regard to infection times, \(T'\) is also a tree (and not a forest).
By definition, there are at least \(\sqrt{n}\) nodes in \(T'\) (examine the path on which \(v\) lays).
Also observe that if the algorithm picks a node from \(T'\) in step 1, then it also picks a node from \(T'\) in step 2, giving us the result for the first part.
Lastly, let us examine how likely it is that a node from \(T'\) is picked in step 1.
For that, see that the probability of the complementary event is at most
$\left((n - \sqrt{n})/n \right)^{\sqrt{n}}$,
since there are at most \(n - \sqrt n\) nodes not in \(T'\) and we pick \(\sqrt n\) nodes. It’s well known that this term approaches \(1/e\) as \(n\)  approaches \(\infty\).
\end{appendixproof}

\begin{appendixtheoremrep}
Let \(0 < \varepsilon < 1 \) be arbitrarily small.
For any algorithm solving the SD problem on trees with consistent source behavior, unknown static graph and watching a single node with success probability \(p >0\), there is an instance for every number of nodes \(n\) such that the price of detection under that algorithm is at least \((1/2) \cdot \sqrt{\ln\left((1-p/(\funnyconstant))^{-2}\right)} n \sqrt{n}\).
\label{thm:aY1-lb}
\end{appendixtheoremrep}

This also gives us a bound on the trade-off between the success probability and the price of detection.
For a more intuitive bound, obverse that by simple analytic tools
% , since \(2 < \funnyconstant < 3\), then we have that \(\sqrt{\ln\left(\frac{1}{(1-p/(\funnyconstant))^2}\right)} > p\) for \(p \in (0,1)\).
% Meaning 
we can deduce that even a price of detection of \(p n\sqrt{n}\) is unachievable for any algorithm with success probability \(p<0.8724\).
% Meaning, if we settle for a decrease in the success probability, we get less than a proportional improvement in the number of tolerated infections.
\recheckNico{I am confused about what you really want to say here. Doesn't \(\sqrt{\ln\left(\frac{1}{(1-p/(\funnyconstant))^2}\right)} > p\) for \(p \in (0,1)\) give you that reducing \(p\) gives you a better than proportional payoff? Or not even that.}
\begin{proofsketch}
For \(n \in \N^+\), we consider the path on \(n\) nodes labeled \(\edgeLabel{i,i+1} \coloneqq i\) for \(i \in [n-1]\).
We assume that there is an algorithm \(A\) that wins the game on these graphs within \(c \sqrt{n}\) rounds with probability at least \(p\).
From \(A\) we construct algorithm \(B\) with strictly stronger capabilities: in each round, \(B\) watches the same nodes as \(A\) and also all nodes that \(A\) has observed to become infected in previous rounds.
We argue that for \(A\) to win the game, \(B\) has to win the game, which happens if \(A\) guesses a node which is close to the source.
We bound the probability of that event to obtain the result.
\end{proofsketch}

\begin{appendixproof}
First, we describe an infinite family of graphs, then assume there in an algorithm winning the game on all of these instances \recheckNico{what does it mean to win a game on a family of graphs?} in \(o(n \sqrt{n})\) infections with probability at least \(p\). We then derive an algorithm that must solve the problem strictly faster and show that, for all small enough \(c\) (which we will give later), there are problem instances where this improved algorithm has probability less than \(p\) to require less than \(c\sqrt{n}\) rounds.

First of all, let \(n \in \N\). Then set \(P_n\) the path with \(n\) nodes and let the time labeling be \(\edgeLabel{i,i+1} = i\) for \(i \in [n-1]\). We shuffle the node indices uniformly at random before providing them to the Discoverer so that they do not reveal the graph For simplicity we will refer to the ordered indices in the rest of the proof. Now let \(A\) be any algorithm that solves the SD problem on the set \(\{P_n\}_{n \in \N}\). w.l.o.g., we may assume that there is a round in which \(A\) watches the root node (if there is no such round, we can modify \(A\) such that in a single extra round before submitting its answer, \(A\) watches the node it will claim is the source).\recheckNico{Maybe we should mention that we shuffle the nodes so that the indices of them do not give anything away about the position? That is somewhat obvious but it might help to mention it explicitly.}

From this, we construct an algorithm \(B\) which has strictly more operations available to it. Concretely, \(B\) may watch an arbitrary number of nodes. In any round, \(B\) chooses to watch the same node \(A\) watches, as well as all nodes that \(A\) has observed being infected. \(B\) then terminates once it has watched the source (the algorithm can tell because the node gets infected but not via an edge). Since we assumed \(A\) watches the node it believes to be the source before terminating, if \(A\) wins the SD game, so does \(B\). Also observe that \(B\) requires at most as many rounds as \(A\).

Observe that, since the players do not know the underlying graph, \(A\) can either select a node about which it has information or one about which it does not have information.
If \(A\) has information about a node, that means it has been picked before or is the source of an infection detected at another node.
Since the behavior of the source is always the same, the first of these two options yields no extra information, and thus we can assume \(A\) never makes such a choice.

By assumption, \(A\) requires less than \(c \sqrt{n}\) rounds (with probability at least \(p\) and for large enough \(n\)). 
We may also assume \(A\) spends at least one round watching the source once it has found it.
Thus, in order for \(B\) to have picked the source at the end, \(A\) must have picked a node with distance at most \(c \sqrt{n }\) from the source at some time \recheckNico{That sentence is confusing. A must have actually picked the source itself, so where does the distance come from?}.
In order for that to happen, \(A\)  must have picked a node about which it had no information and which has distance at most \(c \sqrt{n}\) from the source (by counting the rounds needed to trace an infected node to the source).

Since the nodes about which \(A\) has no information were shuffled uniformly at random by the adversary, we can assume that \(A\) selects them uniformly at random without replacement. Now, since \(\funnyconstant > 2\), the probability of picking a node within distance \(c \sqrt{n}\) to the source is smaller than
$$
1-\left( \frac{n-(\funnyconstant) c\sqrt{n}}{n}\right)^{c \sqrt{n}}.
$$
If we now pick  \(c< (1/2) \cdot \sqrt{\ln\left(\frac{1}{(1-p/(\funnyconstant))^2}\right)}\) and \(n\) large enough such that \(c\sqrt{n} \ge 1\), we have that this success probability is less than p. This follows since, by choice, \(-2c^2 > \ln(1-p/(\funnyconstant))\) \todo{Are you sure that the 2 is on the correct side here?} and thus \(e^{-2c^2} = \left(e^{-2c} \right)^{c} \ge 1-p/(\funnyconstant)\). As \(\left(1-\frac{(\funnyconstant) c}{\sqrt{n}}\right)^{\sqrt{n}}\) approaches \(e^{-(\funnyconstant) c}\) and \(\frac{n-(\funnyconstant) c\sqrt{n}}{n} = 1- \frac{(\funnyconstant) c}{\sqrt{n}}\), we have the desired result \(1-\left( \frac{n-(\funnyconstant) c\sqrt{n}}{n}\right)^{c \sqrt{n}} < p\) for large enough \(n\). This contradicts the assumption that \(A\) had success probability at least \(p\), proving that no such algorithm may exist.
\end{appendixproof}

\begin{toappendix}
Note, the algorithm \(B\) can simply be extended to support an algorithm \(A\) which may watch multiple nodes at the same time. Now, if \(A\) may watch \(k\) nodes, then \(B\) has a success probability of at most
$
1-\left((n-2c\sqrt{n})/n\right)^{kc \sqrt{n}}.
$
\end{toappendix}

\section{When Knowing the Static Graph Helps} \label{sec:known}
\label{known}
One would expect that knowing the static graph where the infections take place puts the Discoverer at a significant advantage.
Surprisingly, we prove that generally, this is not the case.
Motivated by this negative result, we investigate which knowledge about the static graph has value for the Discoverer.
We show that if the static graph has treewidth \(\tw{}\), the price of detection is in \(\BigO(\tw\cdot n \log n)\).
Finally, we argue that the central property that allows this is that the graph then recursively has small separators if it has small treewidth.

\begin{appendixtheoremrep}
Let \(A\) be a Discoverer algorithm for the SD game with consistent source behavior on known static graphs. Then there is an algorithm \(A^u\) for the SD game with consistent source behavior on unknown static graphs such that if \(A\) wins the game within \(O(f(n))\) infections with probability \(g(n)\), then so does \(A^u\) for its game.
\label{thm:consistent-known-to-unknown}
\end{appendixtheoremrep}

\begin{toappendix}
\begin{figure*}[tbph]
\centering
\includegraphics[width=0.8\linewidth,page=2]{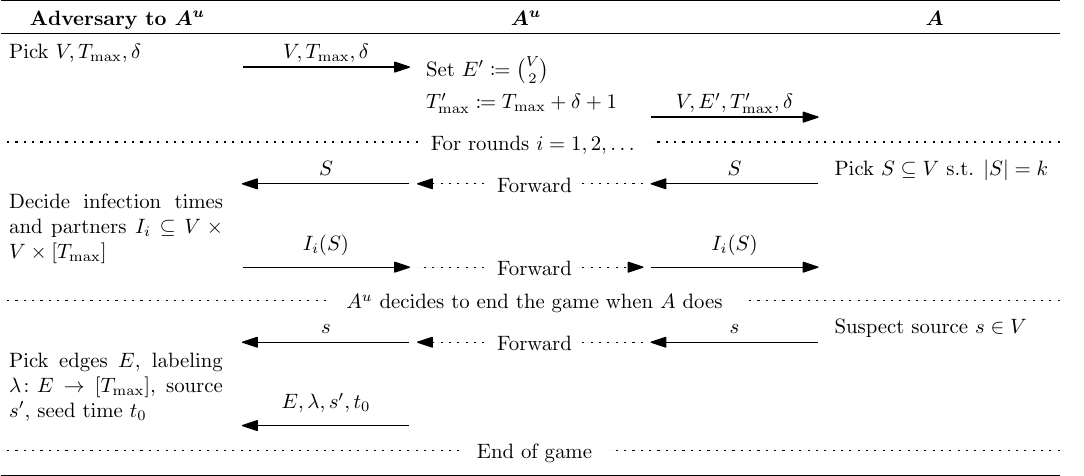}
\caption{Deriving an algorithm \(A^u\) for unknown static graphs from an algorithm \(A\) for known static graphs. \(A^u\) acts as the adversary to \(A\) and uses its behavior to determine its behavior towards its own Adversary. \label{alg:known-to-unknown}}
\end{figure*}
\end{toappendix}

The main idea of this reduction is to use the algorithm \(A\) as a subroutine by reporting \({V \choose 2}\) as the edge set to it.

\begin{appendixproof}
We construct \(A^u\) from \(A\) by simulating the Adversary for \(A\) and using \(A\)'s responses to talk to \(A^u\)'s Adversary.
Let \(V, \Tmax, \iphase\) be the parameters that \(A^u\) receives from the Adversary in the first step of the game.
In the beginning, we pick \(E' \coloneqq \binom{V}{2}\), that is, we report the complete graph to \(A\).
Similarly, we pick \(\Tmax' \coloneqq \Tmax + \iphase + 1\).
We then report these modified parameters to \(A\).
In the rounds phase of the game, we simply relay watching queries and responses between the Adversary and \(A\).
Finally, we simply take \(A\)'s guess of the source and use it as \(A^u\)'s guess.
See \Cref{alg:known-to-unknown} for a diagram of this construction.

Observe that there is a one-to-one correspondence between the infections in the two games.
Thus, if \(A\) finishes within \(O(f(|V|))\) infections for some function \(f \colon \N \to \N^+\), so must \(A^u\).

Define
\begin{align*}
\edgeLabelOp' \coloneqq u,v \mapsto \begin{cases}
\edgeLabel{u,v}, &\text{if } uv \in E, \\
\Tmax', &\text{otherwise.}
\end{cases}
\end{align*}
Consider \Cref{fig:consistent-known-to-unknown} for an illustration of this construction.
Observe that if \(E, \edgeLabelOp, s', t_0\) is consistent with all \(I_i\), then so is \(E', \edgeLabelOp, s', t_0\).
Thus, if the Adversary wins against \(A^u\), the simulated Adversary wins against \(A\).
By contraposition, we may assume that if \(A\) wins against the simulated Adversary, so does \(A^u\) win against its Adversary, yielding the desired translation of success probabilities from \(A\) to \(A^u\).
\end{appendixproof}

\begin{figure}[t]
\centering
\includegraphics[width=5cm]{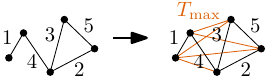}
\caption{\label{fig:consistent-known-to-unknown}The construction from the proof of \Cref{thm:consistent-known-to-unknown}. Discoverer for known graphs \(A\) is used on the completed graph on the right to win the game in the setting of unknown graphs.}
\end{figure}\recheckNico{It should be \(T_{max}'\) right?}

After this negative result, we explore when knowing the static graph does help.
In particular, this also gives us a structural insight into which kinds of static graphs are easy to discover the time labels on.

\begin{appendixtheoremrep}
There is an algorithm that wins the SD game with consistent source behavior on known graphs while only tolerating \(\BigO(\tw \cdot n \log n)\) infections in expectation, where \(\tw\) is the treewidth of the graph.
\label{thm:aX1-tw}
\end{appendixtheoremrep}

% Note that by Markov's inequality, this asymptotic bound also holds with constant probability.\recheckNico{Isn't an expectation result stronger here because the wcp follows from Markov anyways?}

\begin{figure}[t]
\centering
\includegraphics[width=4cm]{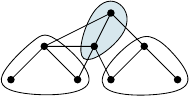}
\caption{\label{fig:aX1-tw}A graph with a balanced separator of size 2 (blue). An infection chain that includes nodes from both the left and right partitions must pass through the separator. Thus, in these cases, watching the separator reveals which of the partitions the source is in.}
\end{figure}

This result is achieved by \Cref{alg:aX1-tw}.
It crucially depends on the existence of small separators.
Specifically, in a static graph \(G\) with treewidth \(\tw{}\), there is a \(1/2\)-balanced separator \(S\) (i.e.,~every connected component in \(G-S\) has size at most \(1/2 \nodesetsize{G}\)) of size at most \(\tw + 1\) \cite{cygan2015parameterized}.
See \Cref{fig:aX1-tw} for an illustration.
This result also holds if the nodes are weighted, and we ask for a separator to ensure each component has at most half of the weight.
Therefore, if we assign either \(0\) or \(1\) as weights, we may pick which nodes we want to evenly distribute on the sides of the separator.

\begin{algorithm}[tbhp]
    \DontPrintSemicolon
    \SetKwFunction{FMain}{TreewidthSourceDetection}
    \SetKwProg{Fn}{fun}{:}{}
       Maintain a set of nodes that could still be the source.

       Until the watched node is the source, repeat:

       1. Compute a balanced \( (\tw+1)\)-size separator of the candidate nodes.

       2. Watch the separator nodes for one round each.

       3. If none of them get infected, watch one of the candidate nodes uniformly at random until one is infected. \label{alg:aX1-tw:random-search}

       4. Update the candidate nodes to only include nodes on the component induced by removing the separator that first infected a node in the separator.\todo{rephrase that sentence. It does not account for case 3 yet.}

\caption{The SD algorithm for static graphs with bounded treewidth and consistent source behavior.}
\label{alg:aX1-tw}
\end{algorithm}

\begin{proofsketch}
The claimed properties of the algorithm follow, since the following three properties hold after each iteration:
\begin{enumerate}[noitemsep]
\item We correctly track the candidate nodes (i.e.,~the source must always be one of the remaining candidates).
\item After each iteration of the main loop of the algorithm, the number of nodes halves.
\item With constant probability, we only tolerate a linear number of infections until detection.\qedhere
\end{enumerate}
\end{proofsketch}
\todo{you say that they all hold after each iteration, but they are not really phrased like that.}

\begin{appendixproof}
The claimed properties of the algorithm follow, since the following three properties hold after each iteration of the main loop:
\begin{enumerate}
\item We correctly track the candidate nodes (i.e.,~the source must always be one of the remaining candidates).
\item After each iteration of the main loop of the algorithm, the number of nodes halves.
\item With constant probability, we only tolerate a linear number of infections until detection.
\end{enumerate}\todo{you say that they all hold after each iteration, but they are not really phrased like that.}

The first property holds inductively.
At the beginning, clearly, the source is one of the candidates since all nodes are candidates.
Then, after each loop, exactly one of three things must have happened: (a) the source was one of the separator nodes; (b) the source is not part of the separator but infected at least one node in the separator; or (c) the source infected no node in the separator, and we found an infected node via the search in step 3.
In case (a), we will have found the source, and the Discoverer wins the game.
In case (b), the source must be in the component of the separated graph from which the infection first entered the separator.
In case (c), the source must be in the same component as the infected node, since if it were not, the infection chain must have included at least one node of the separator.
Thus, the source is always one of the candidate nodes.

The second property simply follows since the separator balances the candidate nodes.
That is, every component induced by removing the separator has at most half of the candidate nodes.
As we have seen, after each round, we select one of these components to restrict the candidate nodes to.

The third property follows from applying \Cref{lem:linear-discovery} to each iteration of the loop, then summing the costs, and finally applying the linearity of expectation.

From these three properties, we see that the main loop runs at most \(\log n\) times and that each of its iterations incurs \(O(\tw \cdot n)\) infections in expectation.
\end{appendixproof}

Intuitively, this dependence on the existence of separators makes sense.
Looking for the source means that in each round, we have to decide where to look next based on the information gleamed so far.
Together, \Cref{thm:aY1} and \Cref{thm:aY1-lb} show that in the case of unknown graphs, we cannot do much better than testing a few nodes and then retracing their infection path node-by-node.
In some cases, we can do better by performing a binary search for the source, but in order to decide which half of our candidate nodes infections stem from, we have to separate them.
Such separators may not be too large since we need to spend rounds on watching them, thus this technique is only useful if small separators exist recursively.
This explains the connection between small treewidth and a smaller price of detection.

\section{The Special Case of Trees} \label{sec:trees}
\label{trees}
First, observe that our algorithm exploiting the treewidth of a known graph (see \Cref{thm:aX1-tw}) directly translates to trees.

\begin{appendixcorollaryrep}
There is an algorithm that wins the SD game with consistent source behavior on known trees while only tolerating \(\BigO(n \log n)\) infections with constant probability.
\label{cor:aX1-tree}
\end{appendixcorollaryrep}

\begin{appendixproof}
This follows from the \Cref{thm:aX1-tw} since trees have treewidth 1.
\end{appendixproof}

For trees, we also show a strong matching lower bound.

\begin{theorem}
There is an infinite family of trees \(\{P_i\}_{i \in \N}\) such that, for any algorithm winning the SD game with consistent source behavior and known static graph on tree graphs with success probability \(p >0\),  the price of detection under that algorithm is asymptotically equivalent to \(n \log n\).
\label{thm:aX1-lb-tree}
\end{theorem}

This implies that there is no algorithm winning this game setup in \(o(n \log n)\) with constant probability.
Again, the proof can be adjusted to instead hold for algorithms that always win the game and have at least a constant probability to finish with less than \(o(n \log n)\) infections.

Surprisingly, the concrete success probability does not play a role for the result (as long as it is positive and constant). Meaning that, for large enough graphs, we cannot trade a lower success probability for a multiplicative improvement over the \(n \log n\) number of tolerated infections.
Compare that to the weaker lower bound from  \Cref{thm:aY1-lb}, where this door is left open (though there is no known algorithm to exploit it).

\begin{proof}
First of all, let \(n \in \N\). Then set \(P_n\) as the path with \(n\) nodes. Set \(\iphase = n\). Now let \(A\) be any algorithm that wins the SD game with consistent source behavior on known trees in \(o(n \log n)\) rounds with a constant probability \(p\). The adversary picks the source node \(s\) uniformly at random. Then, let the adversary set the edge labels for all \(i < s\) as \(\edgeLabel{i,i+1} = n - i\) and for all \(i \ge s\) as \(\edgeLabel{i,i+1} = i\).

Note that we require the algorithm to solve the problem on any tree and with any source with constant probability. We now model all possible randomized algorithms as a distribution over what we call execution trees. These execution trees encode the reactions the algorithm may have to all possible responses by the adversary.  An instance of the game is then equivalent to picking an execution tree at the beginning and evaluating the responses of the adversary against the execution tree, yielding a path through the execution tree.

The behavior of the Discoverer can be described as the series of nodes it picks to watch. We may assume w.l.o.g. that the last node the Discoverer picks is the one it believes is the root (this is similar to the proof of \Cref{thm:aY1-lb}). Now, by the definition of the SD game, in each round, the Discoverer submits a node to watch and receives information about when and from which direction it was infected (if at all). By definition of our instances, every node is infected in every round. Thus, the only information the Discoverer receives is from which direction a node was infected and at which time step.

The edge labels only encode whether the edge is on the left or right side of the source, which the Discoverer also learns because it learns the direction along which the infection travels over the edge.
Thus, we may disregard the effect that this knowledge has on the behavior of the Discoverer.

We construct this (binary) tree as follows. As the Discoverer algorithm is randomized, which node it picks in the next rounds, is dependent on previous information and random choices. We can thus assume the Discoverer makes all its random choices at the start of the execution and then builds a tree of possible outcomes. Each tree node is labeled with the node the Discoverer will watch in a certain round. Then, the left and right children are the choices, the Discoverer algorithm will make if the watched node is infected by its left and right neighbor respectively.

Thus, after picking an execution tree, the behavior of the Discoverer algorithm is deterministic. To complete the result, we show that for any possible execution tree with height in \(o(\log n)\), there are a sufficient amount of source nodes for which the algorithm would not win.
Clearly, an execution of the game is now associated with the execution tree the Discoverer algorithm picks at the start and the path through that tree. Since we assumed the algorithm always watches the source node in at least one round, an execution is only winning if the path taken includes the source nodes.

Assume the algorithm terminates after \(r\) rounds and let \(c \in \R^+\) such that \(r = c \log n\).
We may assume that the execution path has length at most \(r\).
Since the Discoverer can now only win if the source is in the \(r\) first layers of the execution tree (in which there are \(2^r\) nodes) and the source was picked uniformly at random, we have that \(2^r / n \ge p\).
Thus, \(2^{c \log n} / n \ge p\), which rearranges to \(n^c/n \ge p\).
Applying the logarithm and rearranging leaves us with
$c \ge 1 + \log p / \log n$.

The right-hand side approaches \(1\) as \(n\) approaches \(\infty\).
Thus, \(r\) approaches \(\log n\) and since on our family of graphs, all nodes are infected in each round, the number of tolerated infections is asymptotically equivalent to \(n \log n\).
\end{proof}

\section{Dynamic Infection Behavior} \label{sec:dyn}
\label{dyn}
In this variation of the game, the Adversary may change the time at which the source is seed-infected between rounds, but not the source itself.
The Discoverer still learns when the watched node becomes infected each round.
Note that the Adversary does not learn which nodes the Discoverer watches, and so its behavior may not depend on the choices of the Discoverer.
We will see that for this variation, the simple deterministic brute-force algorithm is asymptotically optimal.

\begin{appendixtheoremrep}
The brute-force Discoverer wins any instance of the SD game on unknown graphs and with obliviously dynamic source behavior that tolerates at most \(n^2\) infections.
\label{thm:dY1}
\end{appendixtheoremrep}

This result is not randomized; the algorithm is always successful.
As this variation is the most general, the result directly translates to all the other game variations.
The result is indeed tight, as we see in the following theorem.
We cannot even improve upon the quadratic cost of detection if we only ask for an algorithm with a constant probability of success.

\begin{appendixproof}
The algorithm simply spends one round watching each node.
Since watching the source reveals it, we definitely find the source.
In each round, at most \(n\) nodes become infected, and we require \(n\) rounds.
Thus, we tolerate at most \(n^2\) infections.
\end{appendixproof}

\begin{theorem}
For any algorithm solving the SD problem with obliviously dynamic behavior, known static graph, and watching a single node with success probability \(p >0\), there is an infinite family of graphs such that the price of detection under that algorithm is at least \(\min(1/3, p/2) n^2\).
\label{thm:dX1-lb}
\end{theorem}

\begin{figure}[t]
\centering
\includegraphics[width=3cm]{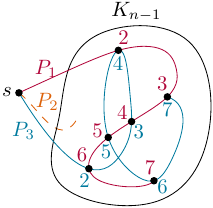}
\caption{\label{fig:dX1-lb}The construction from the proof of \Cref{thm:dX1-lb}. The source is \(s\)  and \(P_1, P_2\) and \(P_3\) are Hamiltonian paths. For simplicity, \(P_2\) is omitted. Nodes are labeled by their position on \(P_1\) and \(P_3\).}
\end{figure}

\begin{proof}
Let \(n \in 2\N+1\).
We construct \(G\) as follows: take \((n-1)/2\) edge-disjoint Hamiltonian paths \(P_1, \dots, P_{(n-1)/2}\) on the \(K_{n-1}\) (which exist by \cite{axiotis_approx}), add a node \(s\) and connect \(s\) to the first node on each of the Hamiltonian paths.
From now on, we interpret each \(P_i\) as starting at \(s\).
Order the nodes and the Hamiltonian paths uniformly at random.
Note, these cover all edges in \(\G\).
Write \(p \colon \edgeset{\G} \to [n/2]\) be the function that maps an edge to the Hamiltonian path it belongs to and write \(q \colon \edgeset{\G} \to [n]\) for the function that assigns an edge its index on that path.
Then let the Adversary act according to \(\edgeLabelOp \coloneqq e \mapsto p(e)n + q(e),\) and pick \(\iphase = 1\).
See \Cref{fig:dX1-lb} for an illustration of this construction.
In each round \(i\), the Adversary lets the source infect the whole graph via \(P_{i \modOp (n-1)/2}\).
For the sake of this proof, we may assume the Discoverer knows the scheme according to which these paths are picked (but not in which order the nodes appear on a specific path).
Thus, if the Discoverer knows the label of an edge \(e\), it may deduce which path it belongs to (by calculating \(\lfloor \edgeLabel{e} / n \rfloor\)) and what the index of the edge has on the path (by calculating \(\edgeLabel{e} \modOp n\)).

We only consider the first \((n-1)/2 - 1\) rounds, since after these, there will have been \(n((n-1)/2-1) \in \Omega(n^2)\) infections.
This is where the \(1/3\) in the minimum in this theorem comes from, since for large enough \(n\), we have \(n((n-1)/2-1) > 1/3 \cdot n^2\).
This case handled \todo{I think it is not really clear what this is an edge case for. Maybe do the case distinction explicitly?}, we assume for the rest of this proof, that during our game there is exactly one infection chain per path.

Intuitively, the Discoverer has information about some of the edges for each path, and the goal of the Discoverer is to find the start of any of the paths.
Formally, we prove that the Discoverer must have observed an infection involving the source (i.e., the first node on all paths), that is, it must have watched the first or second node on a path.
By construction, an observed infection on one path reveals very limited information about the other paths.
An observed infection reveals the positions of the two involved nodes on the current path.
This only tells the Discoverer that they may not be adjacent on any other path.
Within the first \((n-1)/2-1\) rounds, the Discoverer learns precisely \((n-1)/2 - 1\) of these pieces of information, thus for each node, there are at least \(n/2\) positions on the current path that may be consistent with this information.
Also, after \((n-1)/2-1\) rounds, there are at least two nodes on which the Discoverer has no information.
Thus, if the Discoverer never observes the first or second node on a path, the adversary may rearrange the paths in such a way, that one of the unobserved nodes becomes the source (it picks whichever node the Discoverer does not claim is the source).

Since the positions are distributed uniformly at random, no strategy that the Discoverer uses has a probability larger than \(2/n\) \todo{Shouldn't this be a \(4\)? You get a 2 because there are 2 positions and then another 2 because there are at least \(n/2\) valid positions, right?} of picking the first or second node on the current path, thus advancing the requirements outlined above.
Therefore, any algorithm that has success probability at least \(p\), must have at least probability \(p\) of picking the first or second node on a path.
By union bound, the Discoverer must have spent at least \(np/2\) rounds, that is, tolerated \(n^2 p/2\) infections.
\end{proof}

\begin{appendixtheoremrep}
There is a Discoverer algorithm that wins any instance of the SD game on known tree graphs under obliviously dynamic source behavior while watching at most two nodes per round. This algorithm tolerates \(\BigO(n \log n)\) infections in expectation.
\label{thm:dX2}
\end{appendixtheoremrep}

\begin{proofsketch}
    Similar to \Cref{alg:aX1-tw}, we introduce an algorithm that relies on the existence of small spanners.
Concretely, it is a well known result that a tree has a \emph{centroid decomposition} \cite{centroids}.
A centroid is a node such that its removal disconnects the tree into \todo{at least} two components at most half the size of the original tree.
Since, any tree has a centroid and the two components are again trees, we can decompose a tree into a series of these centroids and the respective split trees. The algorithm now proceeds as follows:
\begin{enumerate}[noitemsep]
\item Maintain a subtree of candidate nodes, and always compute a centroid as a balanced separator. Start with the whole graph as the candidates.
\item In each round: watch the current separator and one node in the current subtree picked uniformly at random.
\item If you receive information about the subtree to go into (either because the separator or the other node was infected), do so.
\begin{enumerate}[noitemsep]
\item If the separator is infected, recurse into the side of the separation from which the infection originated.
\item If the randomly picked node is infected but not the separator, recurse into the side of the separation this node is a part of.\qedhere
\end{enumerate}
\end{enumerate}
\end{proofsketch}

\begin{appendixproof}
For the correctness of the algorithm in the proof sketch, note that infection chains always form a directed, connected subtree of the entire graph.
Thus, if the infection stems from one side of the separation, then the source must be part of that subgraph.
Similarly, if the separator is not infected but a node in one of the subtrees is, then the source must be in that subtree.

For the price of discovery, observe that we have to recurse \(\log \nodesetsize{\G}\) times as the centroid is a balanced separator, and thus each time we do so, the subtree containing the remaining candidate nodes halves in size.
Also, by \Cref{lem:linear-discovery}, at most \(\BigO(n)\) nodes are infected in expectation until we observe an infection at the randomly picked node.
Then we may recurse in step 4.

Thus, in expectation, we tolerate \(\BigO(n \log n)\) infections.
\end{appendixproof}

Note that this is the only result in this paper that depends on the Discoverer's ability to watch more than one node at a time.
As we will see in the next theorem, this can only be necessary if the source has obliviously dynamic behavior.
We currently do not know of a lower bound proving that watching more than one node in these scenarios provides an asymptotic advantage.
Intuitively, the above algorithm exploits this capability to mitigate one of the problems with dynamic source behavior: we cannot trace back along an infection chain, as it is unclear if in a given round we see similar behavior to the previous one.
By watching both the next and previous node to be tested, we can extenuate this issue.\todo{I would be more specific about what actually breakes in the previous proof. I.E. That watching a u.a.r. node does not give information anymore about in which component the infection actually started.}

\begin{appendixtheoremrep}
Let \(A\) be a Discoverer algorithm for the SD game with consistent source behavior (with either known or unknown static graph) while watching \(k\) nodes per round.
Then there is a Discoverer algorithm \(A^1\) such that if \(A\) tolerates at most \(x\) successful infections in some instance of the game, \(A^1\) tolerates at most \(k x\) infections.
\label{thm:aXYk-aXY1}
\end{appendixtheoremrep}

This proves that, asymptotically, the ability to watch two nodes does not help the Discoverer in the setting with consistent source behavior.
Note, this theorem and its proof can trivially be extended to randomized bounds.

\begin{appendixproof}
\(A^1\) simply simulates \(A\) in the following fashion: if \(A\) watches \(\ell \le k\) nodes \(v_1, \dots, v_\ell\) in some round \(i\), then \(A^1\) spends \(\ell\) rounds watching \(v_1, \dots, v_\ell\) individually.
Both the correctness and the price of detection bound follow since the behavior of the source is consistent.
Because of the consistency, \(A^1\) gains precisely the same information in the individual rounds for \(v_1, \dots, v_\ell\) as \(A\) does in the one round for \(v_1, \dots, v_\ell\).
Also, because the source behaves consistently, it infects the same number of nodes in each round.
Since \(A^1\) splits each round of \(A\) into at most \(k\) rounds, the number of tolerated infections increases at most by a factor of \(k\).
\end{appendixproof}

% We do not know if this natural relationship holds when the source behavior is obliviously dynamic.
% Intuitively, our argument does not translate, since we cannot rely on the argument that splitting one round with many watched nodes into many rounds with a single watched node yields the same information and tolerates the same number of infections.

\section{Conclusion} \label{sec:conclusion}
\label{conclusion}
Our work provides an extensive theoretical study of SD under the SIR model in temporal graphs.
We give a precise definition of SD in this setting via an interactive two-player game.
Using randomization, we overcome the worst-case behavior and offer efficient algorithms for many settings, such as on trees and general graphs, under consistent and obliviously dynamic source behavior as well as for known and unknown static graphs.
For all but one of the many settings we investigated, we have matching lower bounds proving our algorithms are asymptotically optimal among all algorithms winning the SD game with constant probability, allowing us to characterize the respective difficulty of these settings.

Our work could naturally be extended by studying SD under the susceptible-infected-susceptible model.
As many of our lower bounds rely on the resistance of recovered nodes, seeing which results translate promises to be insightful.
Closing the remaining gap in \Cref{tbl:results} would require investigating SD on unknown graphs with obliviously dynamic source behavior where the underlying static graph is guaranteed to be a tree.
Finally, the most compelling unanswered question is whether allowing the Discoverer to watch multiple nodes in the setting with obliviously dynamic source behavior allows for asymptotically fast algorithms.
We know this is not the case in the setting with consistent source behavior by \Cref{thm:aXYk-aXY1}.
Yet, while we use the additional capability in \Cref{thm:dX2}, we have no proof it is fundamentally necessary to achieve the better running time.

\FloatBarrier
% \section{Bibliography}
% \label{sec:org2c30a6e}
% \bibliography{papers}
% \listoftodos{}
\newpage

\bibliographystyle{named}
\bibliography{papers}

\listoftodos{}

\newpage
\appendix
% \section*{Supplementary Material for Submission XXX}

\end{document}